\documentclass[10pt,journal,onecolumn,compsoc]{IEEEtran}

\newcommand{\subparagraph}{}
\def\BibTeX{{\rm B\kern-.05em{\sc i\kern-.025em b}\kern-.08emT\kern-.1667em\lower.7ex\hbox{E}\kern-.125emX}}

\usepackage{setspace}

\usepackage[utf8]{inputenc}
\usepackage{todonotes}
\usepackage{multicol}
\usepackage{url}
\usepackage{caption,subcaption}  
\usepackage{booktabs}
\usepackage{pgfplots}
\pgfplotsset{width=10cm,compat=1.9}
\usepackage{adjustbox}
\usepackage{mathtools}

\usepackage{amsthm}
\usepackage{breakurl}

\ifCLASSOPTIONcompsoc
  \usepackage[nocompress]{cite}
\else
  \usepackage{cite}
\fi
\usepackage[small,compact]{titlesec}
\usepackage{wrapfig}
\usepackage[font=small,labelfont=bf]{caption} 

\usepackage{float}
\usepackage{epsfig}
\usepackage[linesnumbered,algoruled,boxed,lined]{algorithm2e}
\makeatletter
\newcommand{\removelatexerror}{\let\@latex@error\@gobble}
\makeatother

\makeatletter 
\g@addto@macro{\@algocf@init}{\SetKwInOut{Parameter}{Parameters}} 
\makeatother

\usepackage{comment}

\usepackage{pgfplots}
\usepackage{tikz}
\usepgfplotslibrary{groupplots}
\pgfplotsset{compat=newest}

\usepackage{mathtools}
\allowdisplaybreaks

\usepackage[subtle, paragraphs=normal ]{savetrees}

\usepackage{amsthm}
\usepackage{mathtools}
\usepackage{amsmath}
\usepackage{amsfonts}
\usepackage{amssymb}
\newtheorem{theorem}{Theorem}
\newtheorem{lemma}[theorem]{Lemma}

\newtheorem{definition}{Definition}

\begin{document}
\interfootnotelinepenalty=10000

\sloppy

\title{
The Hermes BFT for Blockchains}

\author{
\IEEEauthorblockN{Mohammad M. Jalalzai, Chen Feng , Costas Busch, Golden G. Richard III, 
Jianyu Niu }\\
}

\IEEEtitleabstractindextext{%
\begin{abstract}

The performance of partially synchronous BFT-based consensus protocols is highly dependent on the primary node. All participant nodes in the network are blocked until they receive a proposal from the primary node to begin the consensus process. Therefore, an honest but slack node (with limited bandwidth) can adversely affect the performance  when selected as primary. 
{\em Hermes} 
decreases protocol dependency on the primary node  and minimizes  transmission delay induced by the slack primary while keeping low message complexity and latency.

Hermes achieves these performance improvements by relaxing strong BFT agreement (safety) guarantees only for a specific type of Byzantine faults (also called equivocated faults). Interestingly, we show that in Hermes equivocating by a Byzantine primary is {\em unlikely}, {\em expensive} and {\em ineffective}. Therefore, the safety of Hermes is comparable to the general BFT consensus.
We deployed and tested Hermes on $190$ Amazon $EC2$ instances. In these tests, Hermes's performance was comparable to the state-of-the-art BFT protocol for blockchains (when the network size is large) in the absence of slack nodes. Whereas, in the presence of slack nodes 
Hermes outperformed the state-of-the-art BFT protocol by more than $4\times$ in terms of throughput as well as $15\times$ in terms of latency. 
\end{abstract}
}

\maketitle
\IEEEdisplaynontitleabstractindextext
\IEEEpeerreviewmaketitle
\IEEEraisesectionheading{\section{Introduction}\label{Introduction}}

Scaling of a consensus protocol to support a large number of participants in the network is a desired feature. It not only allows more participants to join the network but also improves decentralization \cite{ImpossibilityofDecentralization}.  Additionally, reliability of a consensus protocol also depends on the actual number of faults it can tolerate \cite{Gartner03, Experimental-Evaluation}. In a  BFT-based (Byzantine Fault Tolerant) protocol  if the total number of participating nodes is $n$ then the number of faults that can be tolerated  is bounded by $f < n/3$ \cite{Fischer:1985:IDC:3149.214121}. Therefore, to improve the  fault tolerance a scalable protocol has to be designed so that it can operate well with larger values of $n$. However, the increase in number of nodes generally results in higher message complexity which adversely affects the BFT performance \cite{Castro:1999:PBF:296806.296824,DBLP:conf/ifip114/Vukolic15,Luu:2016:SSP:2976749.2978389}.
Recent works have tried to address the scalability issues in the BFT protocols
\cite{Jalal-Window,DBLP:journals/corr/LiuLKA16a, Guerraoui:2010:NBP:1755913.1755950,SBFT,194906}.



In BFT-based protocols the primary node acts as a serializer for requests. 
The consensus epoch begins as the primary  proposes request batch/block to nodes in the network. The epoch ends after nodes agree on a block and add it to their chain. 
In practice, usually the size of the block proposed by the primary may range from several Kilobytes to several Megabytes.  Therefore, by increasing the number of participants/nodes in the network the primary node  with limited bandwidth has to spend most of its time broadcasting the block to a large number of nodes. 
Moreover, consensus nodes are blocked until they receive the proposal from the primary. Inversely, a primary node will be blocked until it receives at least $n-f$ responses from other nodes before proposing the next block. By designing a protocol that shortens this blocking time, BFT-based consensus protocol performance can be improved.

The Hermes protocol design has two new elements. 
The first design element includes proposing a block to a subset of nodes of size $c$ and having a round of message exchange among this subset  to make sure no request (block) equivocation has taken place by the primary.
This design element helps us to achieve two goals namely: $(1)$ High performance despite slack primary nodes; and $(2)$ Efficient Optimistic Responsiveness. 
The second design element involves removing a round (phase) of message exchange among all nodes (equivalent to first phase of PBFT) from happy case execution (normal mode) and adding a recovery step to the view change mode. This design element helps Hermes to achieve design goal $(3)$ Latency.
In addition, we make use of a design element in \cite{Jala1907:Proteus} in order to achieve design goal $(4)$ Scalability.
This element involves node  communication through a subset of nodes of size $c$ also called the \emph{impetus committee}.\looseness=-1

Although the use of a committee to improve BFT performance has been done previously \cite{Jala1907:Proteus,Algorand, SublinearRBFT, Communication-Complexity}, we use the impetus committee for completely different design goals ($1$ through $3$). Moreover, a novel combination of the above three design elements to improve the overall BFT performance is something that to our best knowledge has not been done before.

Below we present key design goals  of Hermes BFT that have pushed the performance to a next level especially for a blockchain setup.\looseness=-1


\textbf{1: High performance despite slack primary nodes.} Slack nodes are honest nodes that have lower upload bandwidth compared to other nodes in the network. Lower upload bandwidth increases transmission delay (which is the time taken to put packets on the wire/link). On the contrary, prompt nodes have higher upload bandwidth and have lower transmission delay.
In normal BFT based protocols \cite{Hot-stuff,Lamport:1982:BGP:357172.357176,tendermint} the primary has to broadcast a block to all nodes in the network of size $n$. Whereas in Hermes the primary broadcasts a block only to a small subset of nodes of size $c$ (also called impetus nodes). In Hermes the growth of $c$ is sub-linear to $n$, therefore increase in $n$ will not have significant impact on $c$, and hence on performance of Hermes. This leads to another interesting property that we call Efficient Optimistic Responsiveness.

\textbf{2: Efficient Optimistic Responsiveness.} 
Responsiveness is an important property of BFT state machine replication (SMR) protocols \cite{BoundsontheTime, Thurdrella}.
In BFT SMR protocols with  responsiveness the primary drives the protocol towards consensus in a time that depends on actual message delay, regardless of any upper bound on message delivery \cite{Hot-stuff,BoundsontheTime}. A protocol is optimistically responsive if it achieves responsiveness when  additional constrains are met.
Hotsuff 
\cite{Hot-stuff} is optimistically responsive in which the primary has to wait for $n-f$ responses to send next proposal. 
In Hermes during period of synchrony, a correct primary will only wait for $c/2 +1$ responses (lower than $n-f$ responses) to propose next block that will make progress with high probability. After a round of agreement these $c$ nodes (also called impetus committee) if agreed, will forward the proposed block to all nodes in the network.
In case of Hermes,  conditions for optimistic responsiveness include receipt of $c/2 +1$ responses by the primary.
For different values of $c$ chosen in our experiments in Section \ref{section:experiments} the probability of having at least one prompt node among a subset of size $c$ is approximately $1-10^{-9}$ (when the number of prompt nodes are at least $f+1$). This means that with high probability there will be at least one prompt  node in the impetus committee that can forward the block to the regular nodes efficiently. 
Therefore we call Hermes efficient optimistic responsive as  messages in Hermes propagate to nodes with the wire speed comparable to the wire speed of prompt nodes.\looseness=-1

\textbf{3: Latency.}
BFT-based protocols \cite{Castro:1999:PBF:296806.296824,SBFT,BFT-SMART} generally operate in two phases (excluding $Pre-prepare$  phase). 
In the first phase each node receives $2f+1$ responses before moving to the second phase. 
The first phase has two objectives.
First, it guarantees that the request is unique. We achieve this objective in Hermes with high probability through message exchange among a small subset of nodes of size $c$ in $Pre-Proposal$ phase shown in Figure \ref{fig:normalmode}. Second, it guarantees that if a request is 
committed in the second phase by at least one node just before a view change, all other correct nodes will eventually commit this request. 
This means for Hermes the second objective of the first phase is only necessary when there is a view change. To save broadcast and processing latency (while achieving second objective), we remove this phase from happy case mode and instead add an additional recovery phase to view change sub protocol in Hermes.

\textbf{4: Scalability.} 
Hermes is using the impetus committee to maintain performance in the presence of a slack primary and achieve efficient optimistic responsiveness. Therefore, the same impetus committee can be leveraged to improve message complexity as done in Proteus \cite{Jala1907:Proteus}.
Thus, we extend a single round of broadcast message where each node receives $n-f$ responses before committing a block into 
 $Proposal$, $Confirm$ and $Approval$ phases (as show in Figure \ref{fig:normalmode}). 
 This reduces the message complexity from $O(n^2)$ to $O(cn)$. \looseness=-1




 The trade-off for improvements in Hermes is a relaxed safety guarantee (R-safety) in the presence of equivocated faults in which a Byzantine primary proposes multiple blocks for the same height just before a view change. Informally, R-safety can be defined as a block will never be revoked if it is committed by at least $2f+1$ nodes or $f+1$ correct (honest) nodes. Whereas in Strong safety (S-safety) a block will never be revoked if it is committed by at least one correct node.
 Strong safety will not hold in the presence of equivocated faults. By making equivocation difficult we increase the likelihood of Hermes to be S-safe. Moreover we also show that the equivocation faults will not have any affect on clients. That means if a client receives commit approval for a transaction from network, then that transaction will never be revoked. Therefore, the main purpose of equivocation which is double spending attack is not possible. 
 The only side affect of the equivocation faults is the network recovery cost from failure. 
We also show that due to the design of Hermes, equivocating faults are unlikely (the primary as well as the majority of the nodes in the impetus committee have to be Byzantine) and expensive  for Byzantine primary and its acquaintances in the impetus committee (the primary as well as the Byzantine nodes in the impetus committee that have signed for two different blocks at the same height will lose their stake in the network or can be blacklisted). Therefore,  a Byzantine primary finds equivocation unlikely, expensive and has no incentive to perform it. 

\begin{figure*}
    \centering
    \includegraphics[width=15cm,height=4cm]{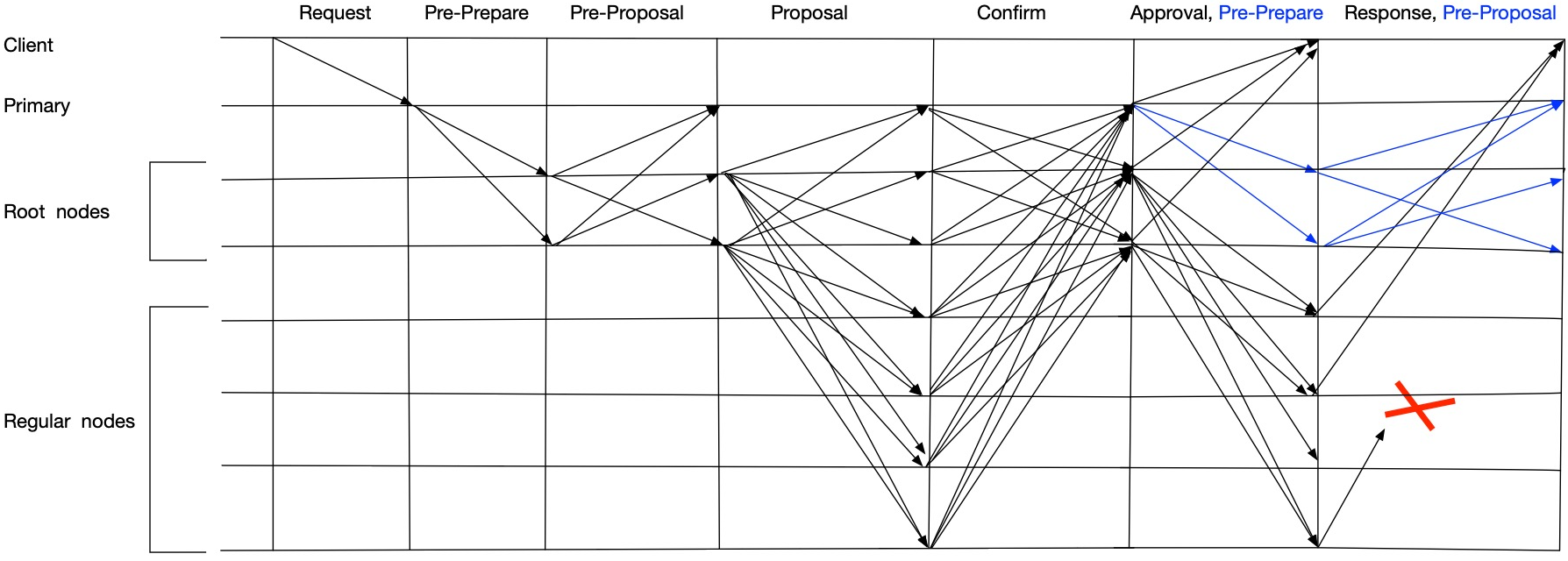}
    \caption{Message pattern in each phase of normal mode in Hermes BFT}
    \label{fig:normalmode}
\end{figure*}

\textbf{Paper Outline.}
The paper is organized as follows.  
In Section \ref{System Model} we give our system model and definitions.
In Section \ref{Protocol} we present the overview of the Hermes protocol. Section \ref{Detailed Protocol Operation} provides a detailed operation of Hermes protocol. 
Proof of correctness for Hermes appears in Section  \ref{Section:Proof of correctness} and in Section \ref{Optimization} we describe  Efficient Optimistic Responsiveness as protocol optimization.
Section \ref{section:experiments} contains the experimental analysis and in Section \ref{Section: Related Work} we present related work .
We conclude our work in Section \ref{Conclusion}.\looseness=-1

\section{Definitions and Model}
\label{System Model}
   Hermes operates under Byzantine fault model.
    Byzantine faults include but not limited to  hardware failures, software bugs, and other malicious behavior. Our protocol can tolerate up to $f$ Byzantine nodes where the total number of nodes in the network is $n$ such that $n=3f+1$. The nodes that follow the protocol are referred to as correct nodes. In this model there can be up to $f$ number of slack nodes ($n_s \leq f$, where $n_s$  is total number of slack nodes ).

    Hermes is a permissioned blockchain. In permissioned blockchain
    nodes in the network are known to each other and have access to each other's public key. In permissioned blockchains nodes can join the network through an access control list.
   In this model nodes are not able to break encryption, signatures and collision resistant hashes. We assume that all messages exchanged among nodes are signed. A message $m'$ signed by the node $i$ is denoted by $\langle m' \rangle_i$, a message $m'$ with aggregated signature from impetus committee quorum (of size $\lfloor c/2 \rfloor + 1$) is denoted by $\langle m' \rangle_{\sigma_r}$, a message $m'$ with aggregated signature from regular node quorum (of size $2f+1$) is denoted by $\langle m' \rangle_{\sigma}$, and  a message $m'$ with aggregated signature from view change quorum (of size $f+1$) is denoted by $\langle m' \rangle_{\sigma_v}$.

  As a state machine replication service Hermes needs to satisfy following properties.


\begin{definition}[Relaxed Safety]
A protocol is R-safe against all Byzantine faults if the following statement holds: in the presence of $\lceil n/3 \rceil -1$ Byzantine nodes, if $ \lfloor 2 n/3 \rfloor$ nodes or  $\lfloor  n/3 \rfloor$ correct nodes commit a block at the sequence (blockchain height) $s$, then no other block will ever be committed at the sequence $s$.
\end{definition}

\begin{definition}[Strong Safety]
A protocol is S-safe if the following statement holds: in the presence of $\lceil n/3 \rceil -1$ Byzantine nodes, if a single correct node commits a block at the sequence (blockchain height) $s$, then no other block will ever be committed at the sequence $s$.
\end{definition}

  \begin{definition}[Liveness]
A protocol is alive if it guarantees progress 
in the presence of 
at most $\lceil n/3 \rceil -1$ Byzantine nodes.
\end{definition}
  Hermes
  guarantees liveness, R-safety as well as probabilistic  S-safety. 
Liveness is achieved in asynchronous network by placing an upper bound on block generation time called timeout.
   Epoch identifies a period of time during which a block is generated. Each epoch can be associated with specific sequence or height of blockchain.
  Once a node receives a block it stops the timer for the expected block and starts its timer 
  for the next block it is expecting from the impetus committee.
  
 For ease of understanding first we describe and prove correctness of the basic Hermes protocol in which the primary proposes the next block once the current proposed block is committed. Then in Section \ref{Optimization} we discuss how Hermes can be optimized so that the primary can propose next block once it collects at least $\lfloor c/2 \rfloor + 1$ responses from the impetus committee (Efficient Optimistic Responsiveness).

\section{Protocol Overview}
\label{Protocol}

\label{protocol}

In Hermes, the primary proposes a block to the impetus committee of size $c$. 
The impetus committee is running an agreement phase  based algorithm
(Algorithm \ref{Algorithm:Customized BFT}).
If $\lfloor c/2 \rfloor + 1$ impetus committee nodes agree on the block proposed by the primary and $2f+1$ regular nodes agree on the proposed block (which they received from the impetus committee), then the block is committed locally by a node and will be added to the blockchain. During normal execution of Hermes we use the same name for a message and its respective phase.\looseness=-1

Normal execution of the protocol is shown in Figure~\ref{fig:normalmode} 
and can be summarized as follows:
\begin{enumerate}
\item 
The primary proposes the block (which contains transactions send by clients) to the impetus committee ($Pre\mbox{-}Prepare$ message).
\item Each impetus committee member verifies the block and makes sure that there is only one block proposed for the expected height
(by collecting signatures from at least $\lfloor c/2 \rfloor +1$ impetus committee members through $Pre\mbox{-}Proposal$ messages).
\item The block is then proposed using the $Proposal$ message which includes  primary signature along with $\lfloor c/2 \rfloor +1$ aggregated signatures of impetus committee members from $Pre\mbox{-}Proposal$ messages. 
\item Upon receipt of a block, regular nodes verify the aggregated signature and the transactions within the block. 
\item If the block is found to be valid, each regular node responds with the signed $Confirm$ message.

\item Upon receipt of $2f+1$ $Confirm$ messages from regular nodes, each impetus committee member as well as the primary
commits the block. Each member of the impetus committee broadcasts  approval of the majority nodes in the form of $2f+1$ aggregated signatures from  regular nodes in $Approval$ message.
\item Upon receipt of approval, each regular node commits the block, which is then added to the local history. \looseness=-1
\item After committing locally, each node sends a $Reply$ message to the client. The client considers its transaction to be committed upon receipt of $2f+1$ $Reply$ messages (or $f+1$ similar $Reply$ messages).
\item Soon after committing a block the primary begins the next epoch (messages are shown in blue) and proposes a new block to the impetus committee using $Pre\mbox{-}Prepare$ message. The impetus committee members begin the $Pre\mbox{-}Proposal$ phase and the protocol progresses through each of its phases as described above.
\end{enumerate}

\subsection{Accountability} \label{subsection:Accountability}

To discourage malicious primary as well as impetus committee nodes from equivocation, Hermes nodes will have to provide a fixed amount of stake while joining the network as done in Proof-of-Stake (PoS) protocols \cite{PoS, tendermint}. But this amount will be equal for every participant in the network. Since every participant has equal stake therefore the probability of being selected as primary node or member of the impetus committee is same for every node. Alternatively, any node involved in equivocation can be blacklisted in the network.\looseness=-1

\subsection{Selecting Impetus Committee Members}
\label{section:impetus selection}

Consider a set of $n$ nodes
and let $\mathcal{N}=\left\{ i |~1\leq i\leq n \right\}$ be their unique node ids, which for simplicity are assumed to be taken between $1$ and $n$. Since all nodes are regular, $\mathcal{N}$ also denotes the set of regular nodes.
Suppose that out of the $n$ nodes at most $f$ 
are faulty such that $f < n/3$;
actually, we will assume the worst case $n = 3f + 1$.

Let $\mathcal{C} \subset \mathcal{N}$ denote the set of nodes in the impetus committee, such that $|\mathcal{C}| = c$,
where $c \ll n$ is a predetermined number 
that specifies the size of the impetus committee (from now on $\mathcal{C}$ and impetus committee will be used interchangeably in the paper).
The impetus committee $\mathcal{C}$ is formed by randomly and uniformly picking 
a set of $c$ nodes out of $n$.
In addition to the impetus committee members, 
there is a primary node picked randomly and uniformly
out of the remaining $n-c$ nodes.


Since $f < n/3$, therefore on expectation, the impetus committee will have less than $c/3$ faulty nodes.
Hence, $\mathcal{C}$ will likely have less than $c/2$ faulty nodes as well.
Nevertheless, as the members of $\mathcal{C}$ are randomly picked, having $c/2$ or more faulty nodes in $\mathcal{C}$ are at long odds.




Moreover, the primary might be faulty as well.
However, a view change can address primary as well as impetus committee failure.
We carry on with analysis to precisely determine the likelihood of different scenarios for picking the members in $\mathcal{C}$.
In the formation of $\mathcal{C}$
the number of possible ways to pick any specific set 
of $c$ nodes out of $n$ is $\binom{n}{c}$.
The probability to pick exactly $a$ correct nodes and $b$ faulty nodes in $\mathcal{C}$, 
such that $a + b = c$, is $
\frac{{\binom{n-f}{a}}{\binom{f}{b}}}{\binom{n}{c}}.$
Therefore, the probability $P_f$ of having at least $c/2$ faulty nodes ($b \geq c/2$) in $\mathcal{C}$ will be:
\begin{equation}
P_f= \sum_{b=\lceil c/2 \rceil}^{c}\frac{{\binom{n-f}{a}}{\binom{f}{b}}}{\binom{n}{c}}.
\label{Eq:Failure Probability Equation}
\end{equation}

If $\mathcal{C}$ is unable to generate a block by the end of timeout period, then $\mathcal{C}$ is replaced by another randomly chosen committee and a new primary through view change.

\textbf{View Change Probability.}
\label{section:view_change_probability}

View change can be triggered either by the failure of impetus committee $\mathcal{C}$ or failure of a primary. More specifically the two  cases that can ultimately result in a view change are: (i) 
$b \geq c/2$, where $b$ is the number of faulty nodes in  $\mathcal{C}$ (this comes with probability $P_f$), or 
(ii)
when $b < c/2$ and the primary node is faulty.
For the latter case ii, since we choose the primary randomly from $n-c$ nodes if the total number of faulty  nodes is $f < n/3$, then the probability of primary being faulty 
is at most $(f - b)/(n-c) < n/(3(n-c))$;
hence, the probability that case ii occurs is less than  $n(1 - P_f)/(3(n-c))$.
Therefore, the probability $P_v$ of having view change due to case i or ii 
is bounded by
    $P_v <  n (1 - P_f)/(3(n-c)) + P_f.$
Since $P_f$ is approaching $0$ and $n \gg c$, 
the upper bound of probability $P_v$ can be approximated by $1/3$ 
asymptotically to the limit of $n$.\looseness=-1

\textbf{Probability of At Least One Correct Node in $\mathcal{C}$.}
For a view change to be initiated, it requires at least one correct node in $\mathcal{C}$.
In the worse case, all the nodes in $\mathcal{C}$ are faulty,
namely, $b = c$; this is the total failure scenario that does not allow view changes.

We observe that the probability of not having any correct node in $\mathcal{C}$ for different values of $n$ and $c$ in our experiments
is at most $3.8 \times 10^{-22}$.
Consequently, the probability of avoiding total failure is extremely high.\looseness=-1

\SetKwFor{Upon}{upon}{do}{end}
\SetKwFor{Check}{check always for}
{then}{end}


\setlength{\textfloatsep}{1.5pt}

\begin{algorithm}[t]
\DontPrintSemicolon 
\setstretch{1}   
\caption{Algorithm for primary node}
\label{Algorithm:Customized BFT}

\If{$i$ is primary} {
\If{There is $\beta$ from previous view}{
\If{$i$ has the payload $m$ for $\beta$}{
Generate block $B$ from $\beta$ with payload $m$\;
}\Else{
Request $Proposal$ for $\beta$ from $\mathcal{C}$\;
\Upon{Receipt of $Proposal$}{
Generate block $B$ from payload $m$ in $Proposal$\;
}
\Upon{Receipt of $2f+1$ negative response}{
Generate block $B$ and also attach $2f+1$ negative responses for $\beta$ \;

}
}
}\Else{
Generate block $B$\;
}
Broadcast $B$ to the set $\mathcal{C}$ \tcp{Pre-Prepare}
}

\Upon{receipt of $2f+1$ valid $Confirm$ messages}{
\tcp{Commit block and increment height}
$\psi_i(B,s) \gets true$ \;
$s=s+1$\;

Send $Response$ to clients\;
}

\end{algorithm}


\begin{algorithm}[h]
\DontPrintSemicolon 
   
\caption{ Node $i \in \mathcal{C}$  }\label{Algorithm:C member}
\Upon{receipt of valid $B$ }{ \label{marker}
Broadcast $Pre\mbox{-}Proposal$ message to $\mathcal{C}$\;
}

\Upon{receipt of $\lfloor c/2 \rfloor +1$ valid $Pre\mbox{-}Proposal$ messages for $B$}{
Build the $Proposal$\;
Broadcast the $Proposal$ message to regular nodes (except the primary) \;
Send $\beta$ to the primary
}
\Check{receipt of a valid $\Gamma_j$ from regular node $j$ or $\Gamma_p$ from primary}{
Execute Algorithm \ref{Algorithm:Synchronization sub Protocol for node e}

}
\Check{for $\Gamma_p$ response}{
Forward the response to the primary
}
\If{not received $B$ by block timeout}{
Broadcast $\Gamma_i$ to regular nodes\;
Accept messages from regular nodes to synchronize local history\;
}

\Upon{receipt of $2f+1$ valid $Confirm$ messages}{
$\psi_i(B,s) \gets true$ \;
Broadcast $Approval$ message\;
Send $Response$ to clients\;
}
\Check{Receipt of first $ConfV$ before receiving $Proof$}{
Broadcast $ConfV$ \;
}


\Check{detecting proof of maliciousness: $E$ complaint or  $f+1$ $\Gamma$ complaints}{
Broadcast $proof$\;
}

\Check{Receipt of 2f+1 ConfV for the view change}{
Do not send Pre-Proposal/Confirm message anymore for this view\;
Build $ApproveV$ from $2f+1$ ConfV messages\;
Broadcast the $ApproveV$ to regular nodes\;
}

\end{algorithm}

\section {The Protocol}
\label{Detailed Protocol Operation}

 The basic operation of our protocol, Hermes, is presented in Algorithms \ref{Algorithm:Customized BFT}, \ref{Algorithm:C member}, and \ref{Algorithm:Regular members}, which describe the normal execution between the impetus committee $\mathcal{C}$ and regular nodes. Algorithm \ref{Algorithm:Synchronization sub Protocol for node e} is executed by synchronization sub-protocol. If normal execution fails, then our protocol switches to view change mode executing Algorithms \ref{Algorithm:Regular ViewChange}, \ref{Algorithm:impetus committee ViewChange} and \ref{Algorithm: New Primary ViewChange} to recover from failure.
Note that the members of $\mathcal{C}$ and the primary also run themselves the protocols for regular nodes in normal mode.\looseness=-1

\subsection{Happy Case Execution} \label{subsection: Normal Mode}
The currently designated primary node $p$ proposes a block by broadcasting a $Pre\mbox{-}Prepare$ message to $\mathcal{C}$ (Algorithm \ref{Algorithm:Customized BFT}, lines 1-9). A $Pre\mbox{-}Prepare$ message from primary $p$ sends a newly created block $B=(\langle ``Pre\mbox{-}Prepare",v,s,h,d,o \rangle_p,m)$ which contains the view number $v$, block sequence number $s$, transaction list $m$, its hash $h$, the previous blockhash $d$ and optional field $o$ which can be used by the primary to send the proof $2f+1$ of negative responses ($F$) during first epoch of its view (new primary in its first epoch might request the latest $Proposal$ from the previous view. If replicas do not have the $Proposal$ they will send negative response $F$). More details about this can be found in the subsection \ref{Subsection: View change}. 
Let $\rho = \langle ``Pre\mbox{-}Prepare",v,s,h,d \rangle_p$.

A node $i$ in $\mathcal{C}$ begins $Pre\mbox{-}Proposal$ phase of the algorithm after receipt of a $Pre\mbox{-}Prepare$ message. Then, node $i$ broadcasts a $Pre\mbox{-}Proposal$ message $\langle ``Pre\mbox{-}Proposal",v,s,h,i \rangle_i$
if it finds the $Pre\mbox{-}Prepare$ message to be valid. The validity check of the $Pre\mbox{-}Prepare$ message includes checking the validity of $s$, $v$, $d$, $h$ and  transactions inside $m$ (Algorithm \ref{Algorithm:C member}, lines 1-3). 
If node $i$ receives $\lfloor c/2 \rfloor + 1$ $Pre\mbox{-}Proposal$ messages from other members of $\mathcal{C}$ for  block $B$ then the node $i$ will successfully create a proposal block $(\langle ``Proposal",v,s,h,d \rangle_{\sigma_r},B)$. This proposal block can be compressed into 
$(\langle ``Proposal",\rho \rangle_{\sigma_r},m)$ and
then node $i$ will broadcast it to the regular committee members; $\sigma_r$ aggregates the signatures of the $\lfloor c/2 \rfloor + 1$ members of $\mathcal{C}$
that contributed to the $Proposal$. Let $\beta = \langle ``Proposal",\rho \rangle_{\sigma_r}$ since the primary already has the payload $m$ (from block $B$ it already proposed), the impetus committee member $i$ will only forward $\beta$ to the primary instead of sending the whole $Proposal$ (Algorithm \ref{Algorithm:C member} lines $4-8$).\looseness=-1 

\begin{figure*}[ht]
\begin{minipage}{0.30\linewidth}
\includegraphics[width=\textwidth, height=3cm]{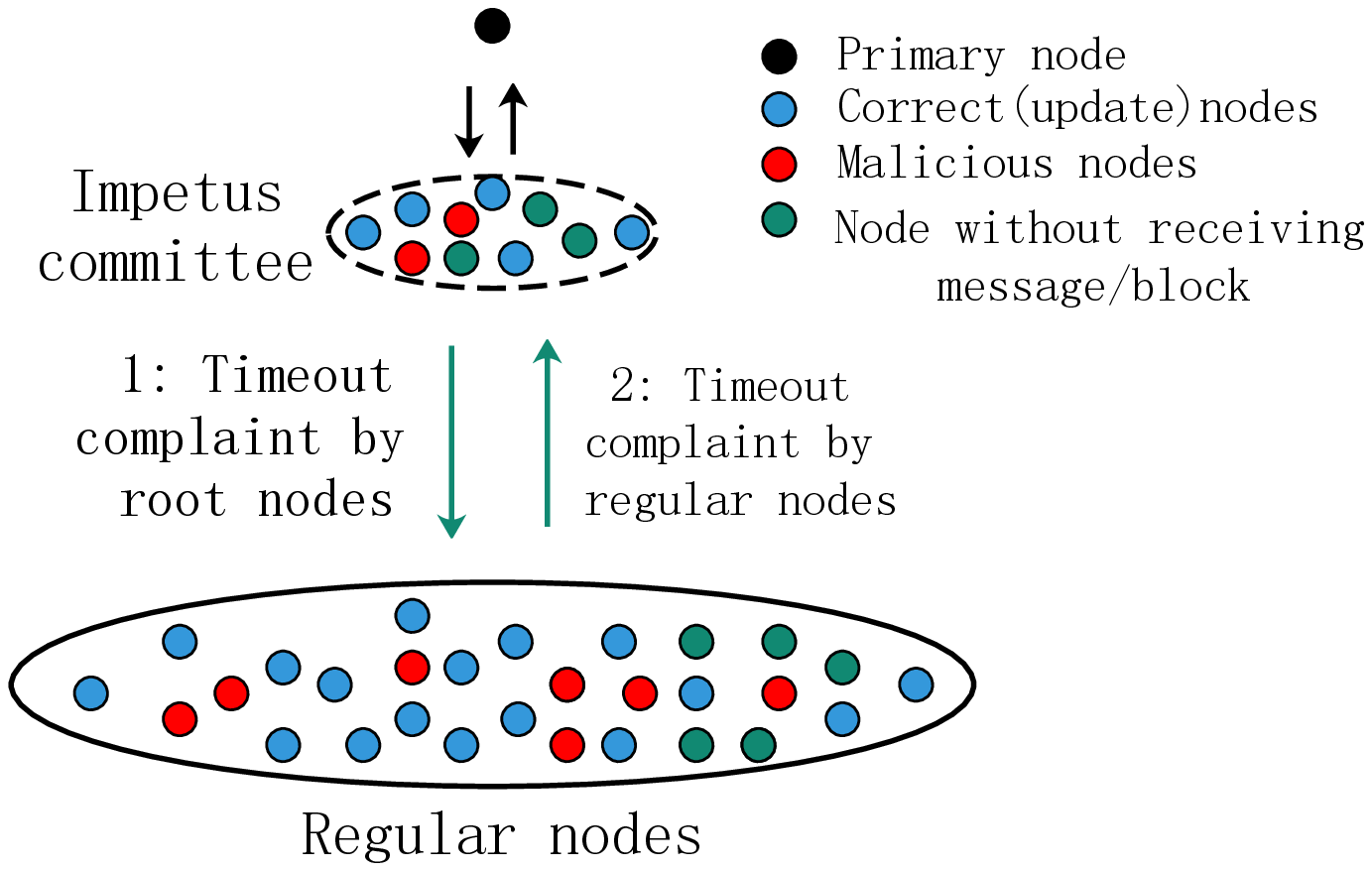}
\caption{Time out complaint }
\label{fig:Synch-Complaint}
\end{minipage}%
\hfill
\begin{minipage}{0.30\linewidth}
\includegraphics[width=\textwidth, height=3cm]{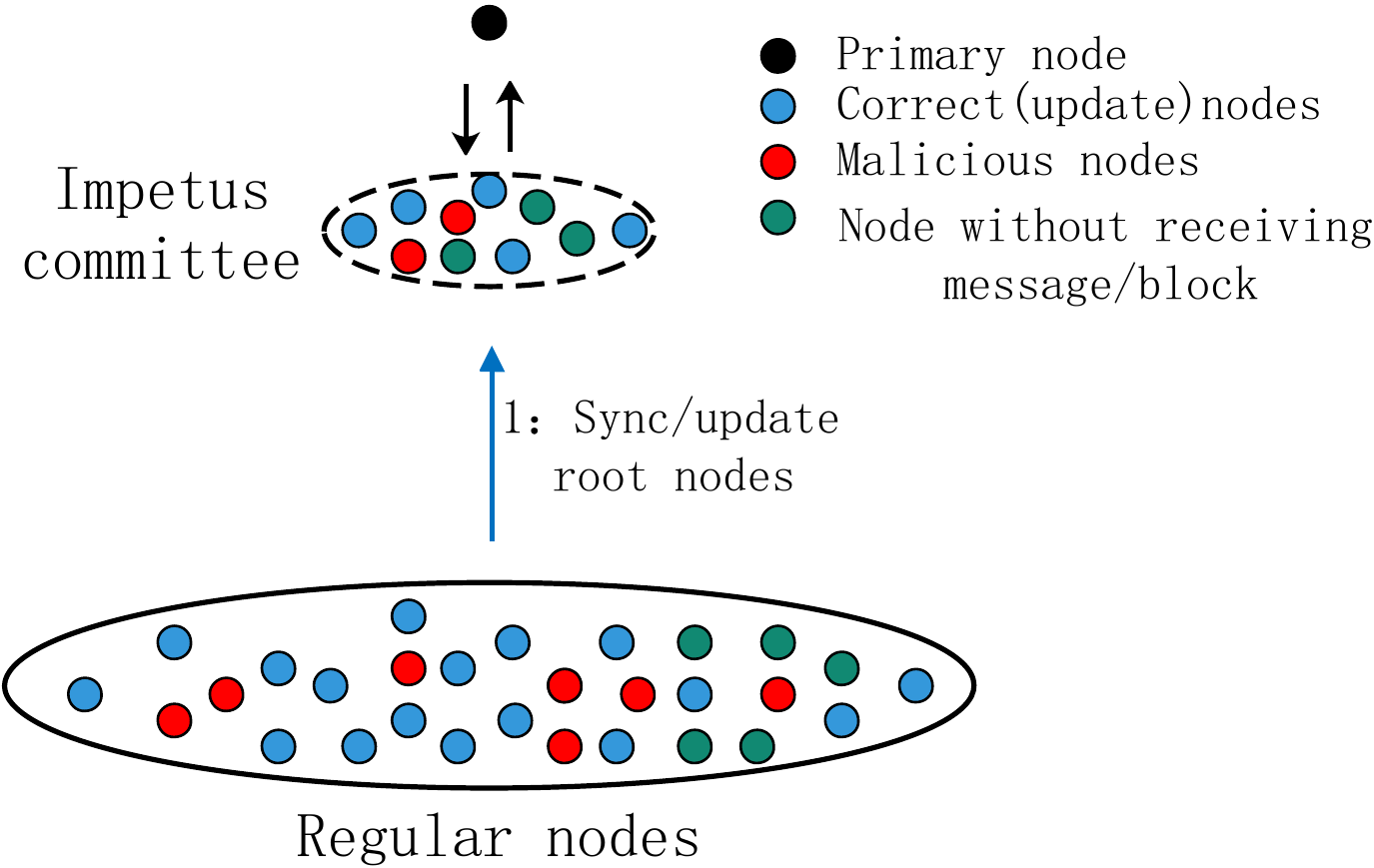}
\vspace{-0.5cm}\caption{Updating $\mathcal{C}$}
\label{fig:impetus-Sync}
\end{minipage}%
\hfill
\begin{minipage}{0.30\linewidth}
\includegraphics[width=\textwidth, height=3cm]{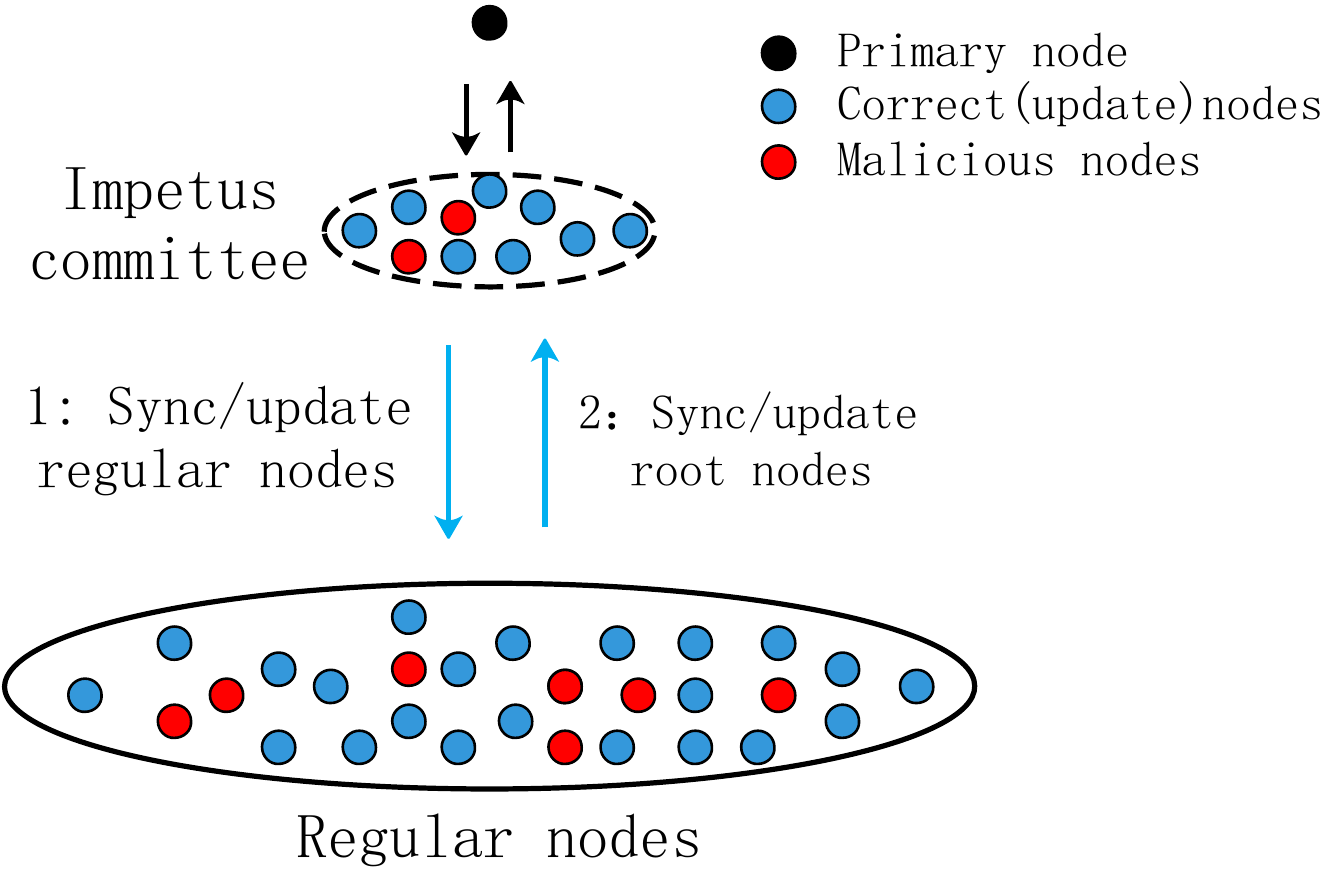}
\caption{Updating regular nodes}
\label{fig:Regular-node-Sync}
\end{minipage}
\end{figure*}

Upon receipt of a $Proposal$ message from $\mathcal{C}$, the regular nodes check if it is signed by at least $\lfloor c/2 \rfloor + 1$ members of $\mathcal{C}$. Regular nodes also verify the block by performing format checks and verification of each transaction against their history. If verification is successful, a regular node $j$ sends back a signed $Confirm$ message $\langle``Confirm",v,s,h,j \rangle_j$ 
to $\mathcal{C}$
(
\ref{Algorithm:Regular members}, lines 1-4). The confirmation of a block $B$ at height $s$ by a node $j$ is denoted by $\eta_j(B,s) \gets true$.

Each member $i$ of $\mathcal{C}$ aggregates $2f+1$ signatures $\sigma$ for $Confirm$ messages 
and then commits the block locally ($\psi_i(B,s)$ is used to show  commit of a block $B$ by a node $i$ at the height $s$). The node $i$ then broadcasts $Approval$ message  $\langle ``Approval", v,s,h\rangle_\sigma$ 
to regular nodes. 
Each member $i$ of $\mathcal{C}$ will also send the $Response$ message  ($\langle  Response, s, v, r, t, i \rangle_{i}$) to each client confirming the execution of the respective transaction (with the timestamp $t$ that the client submitted the transaction and its result $r$)
that it contributed to the block $B$ (Algorithms \ref{Algorithm:Customized BFT}, lines 21-25 and Algorithm \ref{Algorithm:C member} lines $19-23$).
Upon receipt of a valid $Approval$ message from $\mathcal{C}$ (for the current block in the sequence), the regular member $k$ also commits the block ($\psi_k(B,s) \gets true $) as shown in Algorithm \ref{Algorithm:Regular members}, lines 12-21.
Each node $k$ (after committing the block) also sends a $Response$ message as shown in Algorithm \ref{Algorithm:Regular members} line 22.

\begin{algorithm}[t]
\DontPrintSemicolon 
   
\caption{ Synchronization sub Protocol for node e }\label{Algorithm:Synchronization sub Protocol for node e}

\SetKwInOut{Input}{input}
\SetKwInOut{Output}{output}

\Input{$\Gamma$}
\If{$e \in \mathcal{C}$}{
\tcp{If primary requests $Proposal$ from previous view}
\If{$\Gamma_p$}{
\If{have $Proposal$ for $\Gamma_p$ }{
Send $Proposal$ to the Primary
}\Else{
\tcp{Send negative response}
Send $F$ to primary\;
Forward $\Gamma_p$ to regular nodes\;
}
}\If{$\Gamma_j \And j \not \in \mathcal{C}$}{
Update $j$ up to latest locally known block\;
}
}\Else{

\If{$\Gamma_i \And i \in \mathcal{C} $}{
Send missing blocks to committee member $i$
}\If{$\Gamma_p$}{
\If{have $Proposal$}{
 Send $Proposal$ to $\mathcal{C}$
}\Else{
Send $F$ to primary

}
}
}

\end{algorithm}

\begin{algorithm}[t]
\DontPrintSemicolon

\caption{Regular node k}\label{Algorithm:Regular members}

\Upon{receipt of valid  
$Proposal$ (valid  
$\beta$ if k is also primary) from $\mathcal{C}$}{
\tcp{Confirm proposal}
$\eta_j(B,s) \gets true$\;
Generate $Confirm$ message and broadcast it to $\mathcal{C}$\;

}


\If{timeout for a block 
or receipt of invalid block}{
Broadcast $\Gamma \parallel E$  complaint to $\mathcal{C}$\\

\If{sequence number of block is out of order}{
Store the block locally and wait to fill the history gap
}
}
\Upon{receipt $Approval$ message}{
\If {$k \in \mathcal{C}$ or $k$ is primary}{
ignore
}
\Else{
Wait if any blocks were missing from sequence\;
\ForEach{ordered block}{ \tcp{Commit the block at height s}
$\psi_k(B,s) \gets true$ \;
$s=s+1$\;

}
Send $Response$ to clients
}
}

\Check{receipt of a valid $\Gamma_i \parallel \Gamma_p$ from  $i \in \mathcal{C}$}{
Execute Algorithm \ref{Algorithm:Synchronization sub Protocol for node e}
}

\tcp{Initiate view change actions}
\Check{Receipt of $Proof \parallel ConfV$ 
}{ 
Broadcast $ConfV$ to $\mathcal{C}$\;
}
\Check{Receipt of a valid $ApproveV$ for view change}{
\tcp{Transition to new view, based on common random number generation seed}
Randomly select members of $\mathcal{C}$  from $N$\;
\eIf{ $k$ is not $primary$}{
Execute Algorithm \ref{Algorithm:Regular ViewChange} $\parallel$ Algorithm \ref{Algorithm:impetus committee ViewChange}\;
}{
Execute Algorithm \ref{Algorithm: New Primary ViewChange}

}

}

\end{algorithm}

\subsection{Synchronization Sub-Protocol}
\label{subsection:synchronization}
The impetus committee $\mathcal{C}$ might be faulty, when it has $c/2$ or more faulty nodes (shown in red in Figures \ref{fig:Synch-Complaint}, \ref{fig:impetus-Sync}, \ref{fig:Regular-node-Sync}).
In such a case, the faulty members of $\mathcal{C}$ might attempt to not update $\zeta \leq f$ regular nodes without triggering view change; namely, not sending $Proposal$ and other messages to the $\zeta$ nodes.
(In the upper bound of $\zeta$, $f$ may include at most $\lceil c/2 \rceil-1$ failed members of $\mathcal{C}$.)
These $\zeta$ nodes (shown in yellow in Figure \ref{fig:Synch-Complaint}) may not be participating in the consensus process as they have not received messages from members of $\mathcal{C}$ (as majority of $\mathcal{C}$ have failed).
Therefore, they will need to sync their history (download messages) with other nodes. \looseness=-1

Suppose that node $j$ is a regular node that needs to be synchronized ($j \in \zeta $) (download missing blocks). 

Let $i$ be a member of $\mathcal{C}$ such that $i$ has received a timeout complaint $\Gamma_j= \langle ``Timeout",v, s', h',s'',h'' ,\tau,j, \rangle_j $ mentioning that $j$ has not received blocks between the sequences $s'$ and $s''$ with respective hashes $h'$ and $h''$. 
The fields $s''$ and $h''$ are optional and if node $j$ has not received any block after the block at sequence $s'$ then  $\Gamma_j= \langle ``Timeout",v, s', h',\tau,j,\rangle_j$. The message type $\tau$ shows which type of message is missing, e.g. a block, an aggregated view change message $\mathcal{Q}$   (will be discussed along with other message types in subsection \ref{Subsection: View change}). 
The view number $v$ identifies the primary. If node $i$  has not sent a block $B^\star$ to the same node $j$ earlier such that its sequence $s^\star \geq s'$or $s^\star \geq s''$ (validity condition for $\Gamma_j$), then the node $i$ will forward missing messages for missing blocks 
as shown in Figure \ref{fig:Regular-node-Sync}. 

If node $i$ times-out without receiving a valid expected message during the consensus process 
(as shown in the Figure \ref{fig:Synch-Complaint}) it will broadcast a complaint $\Gamma_i$ to regular nodes.
Consequently, a node $i$ in $\mathcal{C}$ can recover a block by receiving it from regular nodes 
(as shown in the Figure \ref{fig:impetus-Sync}). 
Thus, members of $\mathcal{C}$ and regular nodes synchronize their history (download missing blocks from each other) while keeping message complexity low (avoid quadratic message complexity).\looseness=-1 

 When a new primary is elected it has to make sure that if a block is committed by  at least a single correct node then it should be committed by at least $2f+1$ nodes before proposing new blocks.
$\Gamma_p= \langle ``Timeout",v', \beta \rangle_p$ is a specific type of timeout complaint by the new primary to request the latest uncommitted  $Proposal$ (actually the payload $m$ for $\beta$) from previous view. $v'$ denotes the view of the current primary. A negative response from  $2f+1$ nodes for $\Gamma_p$ request will prove that this $Proposal$ has not been committed by any node. If the primary receives a response for $\Gamma_p$ it will re-propose the block. Thus, this additional recovery step after view change makes sure that safety is held after the view change.
More details about the $\Gamma_p$ will be discussed in the subsection \ref{Subsection: View change}.

Algorithm \ref{Algorithm:Synchronization sub Protocol for node e} is about synchronization sub-protocol. Lines 1-14 are executed if the node is a member of the impetus committee. In lines 1-10 node $e$ responds to the $\Gamma_p$ request from primary node. After sending the negative response to the primary, node $e$ forwards $\Gamma_p$ to regular nodes (Algorithm  \ref{Algorithm:Synchronization sub Protocol for node e} line 8). In lines 11-13, the impetus committee member $e$ responds to the request of a regular node $j$ by sending the missing blocks to the node $j$. Lines 15-27 are executed if node $e$ is not a member of the impetus committee.  In lines 16-18, node $e$ responds to an impetus member request/complaint $\Gamma_i$. In lines 19-26, node $e$ responds to request $\Gamma_p$ from the primary node.

The primary node might be malicious 
which can cause the impetus committee $\mathcal{C}$ to fail by not proposing a block.
Another possible cause of failure 
can be the presence of majority of malicious nodes in $\mathcal{C}$. As a result, impetus committee nodes cannot collect at least $c/2+1$ signatures for proposed block. In a rare case, it is also possible that both the primary and the impetus committee $\mathcal{C}$ fail; in such a case the primary can collude with the $\mathcal{C}$ to perform block equivocation. The failure can be detected if $f+1$ nodes send timeout complaint or a single node sends a complaint for equivocation by the primary and $2f+1$ nodes receive it.
Upon detecting failure, Hermes will switch to view change mode to select new primary and impetus committee.

\subsection{View Change}\label{Subsection: View change}
The view change sub protocol in Hermes is different from conventional BFT-based protocols as it has additional recovery phase to make sure if a block is committed by a single node just before view change, then it will be committed eventually by all other correct nodes.
A view change can be triggered if there is sufficient proof of maliciousness of the primary or the impetus committee. The two types of complaint by a node $i$ that can form a proof against the Byzantine primary are the $\Gamma$ complaints (as discussed in subsection \ref{subsection:synchronization}) as well as $E=\langle``Explicit\mbox{-}complaint",v,\epsilon, i \rangle_i$ complaint, 
where $\epsilon$ determines the reason for the complaint which can be either an invalid block $Proposal$ or multiple $Proposal$s by the primary and/or $\mathcal{C}$ with the same sequence number. A $Proof$ is simply formed by $f+1$ valid $\Gamma$ complaints or a single valid $E$ complaint. A single valid $Proof$ is required to trigger a view change.

During each epoch, a regular node waits to receive a proposed block from $\mathcal{C}$. If a regular node $i$ does not receive the block after a timeout then it considers that $\mathcal{C}$ has failed and reports this to $\mathcal{C}$
(Algorithm \ref{Algorithm:Regular members}, lines 5-6).
If $f+1$ nodes report a timeout, then there is at least one correct $\mathcal{C}$ member $j$ that will broadcast the aggregated $f+1$ timeout complaints $(\Gamma)$ to all regular nodes (as $Proof$ in Algorithm \ref{Algorithm:C member}, lines 27-29). Upon receipt of $f+1$ $\Gamma$ complaints the regular nodes send back confirm view change $(\langle``ConfV",\langle v,j\rangle_j,\langle Proof \rangle_{\sigma_v}\rangle)$ message to $\mathcal{C}$, where $\sigma_v$ is the aggregated signature from at least $f+1$ complaints 
from $\Gamma$ messages (Algorithm \ref{Algorithm:Regular members}, line 23-25). 
If complaint is of type $E$ then the confirm view will be of the form $(\langle``ConfV",\langle v,j\rangle_j,\langle Proof \rangle_i\rangle)$, where the node $i$ will be the node that has complained.
If an impetus committee member receives a $ConfV$ message and was not aware of the $Proof$ (therefore did not broadcast it to the regular nodes), it will also broadcast the $ConfV$ to the regular nodes (Algorithm \ref{Algorithm:C member} line 24-25). This step is added to prevent a Byzantine member in $\mathcal{C}$  from  triggering the view change in a subset of regular nodes by sending the $Proof$ messages to few regular nodes.

Once a  member $i$ in $\mathcal{C}$ receives $2f+1$ $ConfV$ messages,
then $i$ triggers view change by broadcasting message 
$(\langle``ApproveV",\langle v\rangle_{\sigma},\langle Proof \rangle_{\sigma_v}, \langle i \rangle_i, \rangle)$ 
to all regular nodes (Algorithm \ref{Algorithm:C member}, lines 30-34). Where $\langle v\rangle_{\sigma}$ is the aggregate signature for $2f+1$ $\langle v,j\rangle_j$ from $ConfV$ messages sent by a regular node $j$.
Upon receipt of the $ApproveV$ message each regular node will begin the view change process (Algorithm \ref{Algorithm:Regular members}, lines 26-33).
In the view change, members of new impetus committee $\mathcal{C}$ along with a new primary is selected by each node using a pre-specified common seed for 
pseudo-random number generation \cite{Luby:1994:PCA:562066}, which guarantees that every node selects the identical $\mathcal{C}$ and primary.
In pseudo-random generator algorithms, a random result can be reproduced using the same seed. For example the count for the epoch trial can be considered as input seed. Count for epoch trial is the total number of successful epochs that generated blocks plus total number of failed epochs that resulted in view change. In Hermes the main requirements for random generator include liveness, bias-resistance and public verifiability or determinism. By using the count for epoch trials as a seed for the the Pseudo-random generator we can satisfy the above mentioned requirements. There are other options like Verifiable Random Functions \cite{VerifiableRF,VRF2} which can also be used in order to provide stronger guarantees like unpredictability.  

After randomly choosing the new primary as well as the new impetus committee $\mathcal{C}$, each node sends a $ViewChange$ message $V_k=(\langle``ViewChange",v+1,s',h,k \rangle_k,\sigma, \beta)$ to the new primary (Algorithm \ref{Algorithm:Regular ViewChange} line 1). 
The $ViewChange$ message has information about the latest block in the local history of the node $k$ ($Approve$ message can be built from $V_k$). It includes the latest committed block sequence number $s'$, block hash $h$, incremented view number $v+1$, and signature evidence of at least $2f+1$ ($\sigma$) nodes that have confirmed the block (through $Confirm$ message). The $\beta$ part of this message includes the latest $Proposal$ and its respective $Pre\mbox{-}Prepare$ message from the last primary (without its payload $m$) which was received by the node $k$ from $\mathcal{C}$ (through block $Proposal$ message). The inclusion of $\beta$ in the view change message is used to \textbf{\textit{recover}} the block (transactions) if less than $f+1$ correct nodes have committed the block. We will discuss more about this at the end of the current section. The information in $V_k$  is used to determine whether nodes have different local histories, 
which in turn can allow them to synchronize their histories by getting the possible missing blocks
from the new members of $\mathcal{C}$. 
Upon receipt of $V_k$ from node $k$, the new primary extracts the parts $h$ (hash of the latest committed block), $s'$, and $k$ and also checks the validity of $\sigma$ and $\beta$. Out of the $2f+1$ nodes that contribute to $\sigma$, 
it is guaranteed that at least $f+1$ are correct nodes and at least one out of these $f+1$ correct nodes has the latest block.

This guarantee is sufficient to generate next blocks in the new epoch on the top of the latest block.

The new primary, once it receives $2f+1$ $V_k$ messages it aggregates 
them into $\mathcal{Q}$ which it broadcasts
to members of $\mathcal{C}$ (Algorithm \ref{Algorithm: New Primary ViewChange}, lines 1-6)
and then $\mathcal{C}$ will broadcast the message $\mathcal{Q}$ to all nodes (Algorithm \ref{Algorithm:impetus committee ViewChange}). 
Upon receipt of $\mathcal{Q}$, node $k$ makes sure that its history matches with the history in $\mathcal{Q}$ (agreed by at least $f+1$ nodes) and if it does, node $k$ sends back a $Ready$ message $R_k=\langle``Ready",v+1,s',h,k\rangle_k$ to the new primary
(Algorithm \ref{Algorithm:Regular ViewChange}, lines 2-8). The primary will aggregate $2f+1$ $Ready$ responses into a single $P$ message and broadcast it to $\mathcal{C}$ (Algorithm \ref{Algorithm: New Primary ViewChange}, lines 11-16)
which will be forwarded to all the nodes (Algorithm \ref{Algorithm:impetus committee ViewChange}). Upon receipt of $P$, node $k$ is now ready to take part in new view (Algorithm \ref{Algorithm:Regular ViewChange}, lines 9-12).  If node $k$'s history does not match that of $\mathcal{Q}$ it will synchronize its history (Algorithm \ref{Algorithm:Regular ViewChange}, lines 4-6).

\begin{algorithm}[!htbp]
\DontPrintSemicolon

\caption{View change for regular node $k$}\label{Algorithm:Regular ViewChange}

Send local history $V_k$ to new primary\;
\Upon{receipt of a valid $\mathcal{Q}$ message from a new committee member $i$}{
Extract the most recent valid history from $\mathcal{Q}$\;
\If{local history is not same as most recent history in $\mathcal{Q}$}{
Synchronize local history according to $\mathcal{Q}$\;
}
Broadcast READY message ($R_k$) to new primary\;


}
\Check{receipt of $P$ from $i \in \mathcal{C}$}{
\If{$P$ has at least $2f+1$ distinct READY messages}{

return to Algorithm \ref{Algorithm:Regular members}\;
}
}
\Check{ $\Gamma_i$ where $i \in \mathcal{C}$}{
Execute Algorithm \ref{Algorithm:Synchronization sub Protocol for node e}

}
\end{algorithm}

\begin{algorithm}[t]
\DontPrintSemicolon

\caption{View change for node $i \in \mathcal{C}$}\label{Algorithm:impetus committee ViewChange}
\If{$i \in \mathcal{C}$}{
\Check{ ($\mathcal{Q} \parallel \mathcal{P}$) \text{ from primary}}
{
Forward ($\mathcal{Q} \parallel \mathcal{P})$ \text{ to regular nodes }\;
\If{$P$ has at least $2f+1$ distinct READY messages}{

return to Algorithm \ref{Algorithm:C member}\;
}
}
}
\Check{ $\Gamma_j$ where $j \not \in \mathcal{C}$}{
Execute Algorithm \ref{Algorithm:Synchronization sub Protocol for node e}

}

\end{algorithm}

\begin{algorithm}[!htbp]
\DontPrintSemicolon

\caption{View change for new  primary}\label{Algorithm: New Primary ViewChange}

\Upon{receipt of $V_i$ from node $i$}{

$\mathcal{Q} \gets  \mathcal{Q} \cup V_i $ \;
\If{$\mathcal{Q}$ contains at least $2f+1$  histories}{
Broadcast $\mathcal{Q}$ to $\mathcal{C}$\;
}

}
Extract the most recent valid history\;
\If{local history is different from most recent history}{
Synchronize local history $V_i$ according to the $\mathcal{Q}$ from regular nodes
}
\Check{receipt of $R_i$}{
$P \gets  P \cup R_i$\;
\If{$P$ has accumulated at least $2f+1$ distinct READY messages}{
Broadcast $P$ to members of $\mathcal{C}$\;
}

}
\Check{ history update request from regular node $i$}{
\If{ node $j$ has the latest history/block and has not already sent it to $i$} {
Send the blocks up to  the latest block to regular node  $i$}
}
Return to Algorithm
\ref{Algorithm:Customized BFT}
\end{algorithm}

\textbf{Recovery During View Change.}

Recall that we add a recovery phase to the view change in order to significantly reduce the latency. During this recovery phase, Hermes protocol makes sure that if there is any block that has been committed by at least one node, it will be committed by all correct nodes eventually.
 If there is $\beta$ in  $V_k \in \mathcal{Q}$, then the new primary has to propose the respective block $B$ for the $\beta$.  
  If the new primary has already received the $Proposal$ or the block $B$ for $\beta$ it will re-propose it in the first epoch of its view.  On the other hand, if the new primary does not have the complete $Proposal$ (including its payload $m$) for $\beta$ it will request it from $\mathcal{C}$ using a $\Gamma_p$ complaint. If an impetus committee member $e$ has the $Proposal$ (payload $m$) it will send it back to the primary. If not, it will forward the $\Gamma_p$ complaint to the regular nodes as well as sending a negative response $F=\langle \beta,false \rangle_e$ to the primary (Algorithm \ref{Algorithm:Synchronization sub Protocol for node e}, lines 1-10). Regular nodes will also send back the $Proposal$ message for $\beta$ if they have it to the impetus committee (which will forward it to the primary) or a negative response $F$ 
to the primary.
 If the primary receives back the $Proposal$ it will re-propose the block 
 else it will have to propose another block with $2f+1$ aggregated signatures for negative response $F$ being attached to it in order to prove that $Proposal$ for $\beta$ was not committed by any correct node (Algorithm \ref{Algorithm:Customized BFT}, lines 1-15). Therefore, if there is a $\beta$ in $V_k \in \mathcal{Q}$ during recent view change, then in the first epoch after the view change the new primary  either has to propose the relative block for  $\beta$ or include $2f+1$ negative responses in the first block $B$ that it proposes in the new view. If the primary fails to do so, a view change will occur (through $E$ complaint).

\section{Analysis and Formal Proofs} 
\label{Section:Proof of correctness}
In this section we provide proofs for R-safety, S-safety and liveness properties for Hermes. 

\subsection{Safety}
Hermes has relaxed the agreement (safety) condition. In Hermes a block is committed if at least $f+1$ correct nodes commit the block. If $2f+1$ nodes commit a block it is guaranteed that at least $f+1$ of these nodes are correct. Due to this relaxation in the agreement condition the client also needs to wait for at least $2f+1$ $Response$ messages for its transaction to consider it committed (to make sure at least $f+1$ correct nodes have committed the block). For further client and network interaction including addressing, client request de-duplication, censorship etc., we would defer to standard literature. 

But it is also highly likely that Hermes holds stronger agreement condition where if a single correct node commits a block all other correct nodes will eventually commit the same block at the same height (due to low probability of the primary as well as the majority of nodes in $\mathcal{C}$ are Byzantine).
For simplicity we use $\mathcal{H}$ as the set of all correct nodes such that $\mathcal{H} \subseteq \mathcal{N}$ and $|\mathcal{H}| \geq 2f+1$ in the proof of correctness.

\begin{lemma}\label{Lemma:Safety-Normal-Mode}
Hermes is R-safe during normal mode.
\end{lemma}
\begin{proof} 
Consider two different blocks $B$ and $B'$ with respective heights $s$ and $s'$
where each block has been committed by at least $2f + 1$ nodes.
It suffices to show that $s \neq s'$. Suppose, for the sake of contradiction, that during happy case execution (normal mode) both $B$ and $B'$ are committed at the same height ($s = s'$). Let $\mathcal{K}_1\subseteq \mathcal{N}$ be the set of nodes that commit $B$
such that $|\mathcal{K}_1| \geq 2f+1$ and 
for each $i \in \mathcal{K}_1$, $\psi_i(B,s) \gets true$ (all members of $\mathcal{K}_1$ have committed $B$ at height $s$). Similarly, let $\mathcal{K}_2 \subseteq \mathcal{N}$ 
be the set of nodes that commit $B'$
such that $|\mathcal{K}_2| \geq 2f+1$ and 
for each $i \in \mathcal{K}_2$, $\psi_i(B',s') \gets true$. 
Since $n \geq 3f+1$ and $f < n/3$, 
we get for $\mathcal{K}_1 \cap \mathcal{K}_2= \mathcal{K}$ that $|\mathcal{K}| \geq f+1$.
Hence, $\mathcal{K} \cap \mathcal{H} \neq \emptyset$. Therefore, there is an $i\in\mathcal{K} \cap \mathcal{H}$ such that $\psi_i(B,s) \gets true$ and $\psi_i(B',s') \gets true$
(that is, $i$ committed both blocks). However, since $i \in \mathcal{H}$, $i$ can only commit one block at any specific sequence, and hence, it is impossible that $i$ executed 
both $\psi_i(B,s) \gets true$ and $\psi_i(B',s') \gets true$ with $s = s'$ and $B \neq B'$. 
Therefore, $s \neq s'$, as needed. \looseness=-1
\end{proof}
\begin{lemma}
\label{Lemma:Safety-ViewChange-Mode}
Hermes is R-safe during view change.
\end{lemma}

\begin{proof}
Consider the latest block $B'$ at height $s'$
that was committed by at least $2f+1$ nodes before the view change.
We will show that $B'$ will be included in the blockchain history after the view change.
Let $\mathcal{H}_c \subseteq \mathcal{H}$ be the set of honest nodes that have committed $B'$ (that is, $\psi_i(B',s') \gets true$ for each $i \in \mathcal{H}_c$).
Since the number of Byzantine nodes is $f$,
we get that $|\mathcal{H}_c| \geq f + 1$.
At least a node $i \in \mathcal{H}_c$ must have reported a view change message $V_i$ 
to the primary, such that $V_i$ is aware of $B'$ and $s'$.  
Since the primary node collects $2f+1$ view change messages into $\mathcal{Q}$, 
we get that $\mathcal{Q}$ contains at least $f+1$ view change messages from the nodes in $\mathcal{H}$ where at least one node is in $\mathcal{H}_c$ too.
Hence, when a node $j$ receives $\mathcal{Q}$ from the new primary,
$j$ will know that $B'$ has  been committed, a guarantee that $B'$ is valid.
Thus, block $B'$ will be inserted into the local history of every node that receives $\mathcal{Q}$, and becomes part of the blockchain history.

\end{proof}

Therefore, we get the following theorem:

\begin{theorem}
Under normal operation and  a view change, Hermes provides relaxed safety guarantees (Hermes is R-safe).
\end{theorem}
\begin{proof}
From Lemmas \ref{Lemma:Safety-Normal-Mode} and \ref{Lemma:Safety-ViewChange-Mode}  we obtain, Hermes provides relaxed safety guarantees (Hermes is R-safe) during normal mode as well as view change mode.
\end{proof}

\begin{lemma} \label{lemma: S-safe normal mode}
Hermes is S-safe during normal mode.
\end{lemma}
\begin{proof}
We will prove this lemma by contradiction. Let us assume that during normal mode at least one correct node (say, node $i$) commits a block $B$ at height $s$ then we have  $\psi_i(B,s) \gets true$. This means there exists a set $\mathcal{K}_1\subseteq \mathcal{N}$ 
such that $|\mathcal{K}_1| \geq 2f+1$ and 
for each $i \in \mathcal{K}_1$,  ($\eta_i(B,s) \gets true$) (all members of $\mathcal{K}_1$ have confirmed $B$ at height $s$). We also assume that there is another node $j$ that has committed a block $B'$ at height $s'$ such that $s'=s$ ($\psi_j(B',s') \gets true$). Therefore, there is another set 
$\mathcal{K}_2 \subseteq \mathcal{N}$ such that $|\mathcal{K}_2| \geq 2f+1$ and 
for each $i \in \mathcal{K}_2$, ($\eta_i(B',s') \gets true$). 
Based on $n \geq 3f+1$ and $f < n/3$, 
we get  $\mathcal{K}_1 \cap \mathcal{K}_2= \mathcal{K}$ such that $|\mathcal{K}| \geq f+1$.
Hence, $\mathcal{K} \cap \mathcal{H} \neq \emptyset$. Therefore, there is an $i\in\mathcal{K} \cap \mathcal{H}$ such that $\eta_i(B,s) \gets true$ and $\eta_i(B',s') \gets true$
(that is, $i$ confirmed both blocks). However, since $i \in \mathcal{H}$, $i$ can only confirm one block at any specific level (height) during normal mode, and hence, it is impossible that $i$ confirmed 
both $\eta_i(B,s) \gets true$ and $\eta_i(B',s') \gets true$ with $s = s'$ and $B \neq B'$. 
Therefore, $s \neq s'$. 
\end{proof}

\begin{lemma}\label{lemma:S-safe without Equivocation}
Hermes is S-safe during view change if there is no equivocation fault.
\end{lemma}
\begin{proof}
Let us assume that a single node $i$ has committed a block $B$ at the height $s$  ($\psi_i(B,s) \gets true$) just before the view change occurs. That means there is a set $\mathcal{K}_1 \subseteq \mathcal{N}$ such that $|\mathcal{K}_1| \geq 2f+1$ and 
for each $i \in \mathcal{K}_1$, ($\eta_i(B,s) \gets true$). Therefore during view change there is another set  $\mathcal{K}_2 \subseteq \mathcal{N}$ such that $|\mathcal{K}_2| \geq 2f+1$ and for each $i \in \mathcal{K}_2$, $V_i \in \mathcal{Q}$.  As $n \geq 3f+1$, therefore we have $\mathcal{K}_1 \cap \mathcal{K}_2= \mathcal{K}$ such that $|\mathcal{K}| \geq f+1$. Thus, $\mathcal{K} \cap \mathcal{H} \neq \emptyset$. Therefore, there is an $i\in\mathcal{K} \cap \mathcal{H}$ such that $\eta_i(B,s) \gets true$ and $V_i \in \mathcal{Q}$. Since $V_i$ contains $\beta$ for block $B$, therefore the new primary has to re-propose block $B$ in the first epoch of the new view in the height $s$ as show in Algorithm \ref{Algorithm:Customized BFT} lines 1-20. 
\end{proof}

\subsection{The equivocation Case}
\label{subsection:equivocation case}
If the previous primary is malicious it can propose multiple blocks at the same height/sequence (just before the view change). Suppose the malicious primary proposes two blocks $B'$ and $B''$ before the view change such that one of these blocks (say, block $B'$) is committed by at  most $\xi$ number of nodes where $\xi \leq f$. In such a case it is not guaranteed that a block committed by $\xi$ nodes can remain in the chain after the view change.

If $2f+1$ view change messages received by the new primary (correct) include  a view change message $V_k$ by the node $k$ that contain the commit certificate ($Approve$ message with $2f+1$ signatures) for the block $B'$ then it is guaranteed that the block $B'$ will be re-proposed in the next epoch for the same height (not be revoked). But if the new primary  collects the view change messages from another $2f+1$ nodes (not the $f$ nodes that committed the block) then the new primary might either propose $B'$ or $B''$  in the next epoch. Therefore it is not guaranteed that the block $B'$ that has been committed by at most $f$ nodes (at the height $h$) will be re-proposed at the same height. This can happen because of the fact that 
out of these $2f+1$ nodes included in the $\mathcal{Q}$ prepared by the new primary at most $f$ nodes can be Byzantine  whereas at most another $f$ might have received the block $B''$. At least one node can have the block $B'$. Therefore the new primary may receive two $\beta$ messages, $\beta'$ (for block  $B'$) and  $\beta''$ (for block  $B''$)  from the nodes in the view change message. The primary can re-propose any one of these blocks ($B'$  or $B''$ for example) in the first epoch of the new view (if found to be valid). If block $B'$ is proposed then it is fine as this block is committed by at most $f$ nodes. But if $B''$  is proposed then the nodes that have committed the block $B'$ will have to revoke the block $B'$ and perform agreement on the block $B''$ at the same height/sequence. Block $B'$ can be proposed in the next epoch if there is no common transaction between block $B'$ and $B''$. If there is at least one common transaction among the blocks $B'$ and $B''$  or some transactions in the block $B'$ are invalid the primary can extract valid transactions from the block $B'$ and propose them in another block. It should be noted that in such case the malicious primary and its culprits in the $\mathcal{C}$ can be punished and their stake in the network can be slashed. \looseness=-1

\begin{lemma}\label{lemma:unlikely equivocation}
It is highly unlikely that a Byzantine primary as well as impetus committee perform equivocation.
\end{lemma}
\begin{proof}
We prove this lemma three points. First the probability of performing equivocation by a Byzantine primary is low. Secondly, it is expensive for a Byzantine primary to perform equivocation and third there is not enough incentive for a Byzantine primary to perform equivocation as it will only cause processing delay in the network to recover from faults without affecting the client.

In order for an equivocation to take place the primary as well as the majority of nodes in the impetus committee $\mathcal{C}$ have to be Byzantine. The probability of such an event with different values of $n$ and $c$ used in Section \ref{section:experiments} is between 
$0.00132$ and $0.0027$. For example if on average the network has a view change once per week it will take $7$ to $14$ years for a primary as well as the $\mathcal{C}$ to be Byzantine and be able to propose multiple blocks at the same height to the network.
Additionally, the primary as well as the members of impetus committee that are involved in equivocation will lose their stake in the network or can be blacklisted as described in subsection \ref{subsection:Accountability}.\looseness=-1

As mentioned in the subsection \ref{subsection:equivocation case}, in case of equivocation if less than $f+1$ nodes have commit a block then this block might get revoked. But it should also be noted that unlike other forking protocols (bitcoin, Ethereum etc) where a client has to wait for the confirmation time to make sure the block it has accepted will not be revoked, this (revoking of block $b$) has no affect on the client because a client will only consider a transaction in a specific block to be committed if and only if at least $2f+1$ nodes in the network commit that block. Thus, if a client considers a transaction to be committed then it has been committed by $2f+1$ nodes or $f+1$ correct nodes and it will not be revoked (as shown in the Lemma \ref{Lemma:Safety-ViewChange-Mode}). Therefore a Byzantine primary has no strong incentive to perform equivocation as revoked transaction is not considered committed by the client (double spending attack is impossible).   
As a result, it is unlikely that equivocation will take place in Hermes.

\end{proof}

\begin{theorem}
It is highly likely that Hermes will hold S-safety.
\end{theorem}
\begin{proof}
This is a direct consequence of Lemmas \ref{lemma: S-safe normal mode}, \ref{lemma:S-safe without Equivocation}
and 
\ref{lemma:unlikely equivocation}.\looseness=-1
\end{proof}

\subsection{Liveness}

Hermes uses different means to provide liveness while keeping message complexity low (at $O(cn)$ messages). As the main communication among the nodes occur through the $\mathcal{C}$, it is important that there must be at least one correct node in the $\mathcal{C}$ to guarantee liveness.
As shown in Section \ref{section:impetus selection}, the largest probability of total failure (for liveness) is $3.8 \cdot 10^{-22}$ considering different practical sizes of $c$ and $n$ that we have used in Section \ref{section:experiments}. \looseness=-1

Another case that could potentially prevent liveness is when repeatedly selecting a bad $\mathcal{C}$ or malicious primary after a view change, which in turn triggers another view change, and this perpetuates without termination.
However, as we show next, this extreme scenario may only occur with extremely low probability that fast approaches to $0$.
If the primary is malicious or the number of malicious nodes in the impetus committee $\mathcal{C}$ is at least $\lfloor c/2 \rfloor + 1$, then based on Algorithm~\ref{Algorithm:C member} (lines 36-40), and Algorithm~\ref{Algorithm:Regular members} (lines 28-38), a view change may occur.
The probability $P_v$ of such a bad event causing a view change is approximately a constant $1/3$
(Section \ref{Protocol}).

Treating each such bad event as a Bernoulli trial, we have that the bad events trigger consecutively $\kappa$ view changes with probability at most $P_v^\kappa$, which fast approaches to $0$ with exponential rate as $\kappa$ increases linearly.
Therefore, the probability of this scenario is negligible and does not affect liveness.

Additionally, view change in Hermes also employ three techniques applied by the PBFT \cite{Castro:1999:PBF:296806.296824}. These techniques include: $(1)$ exponential backoff timer for view change; $(2)$ at least $f+1$ complaints will cause a view change. This 
technique has been customized in Hermes with additional phases as described in \ref{Subsection: View change}. The first  phase that has been added to Hermes include broadcasting $ConfV$ messages to regular nodes by impetus committee. The second phase include aggregation of at least $2f+1$ $ConfV$ messages by the members of $\mathcal{C}$ into  $ApproveV$ message and broadcasting it to regular nodes. These additional steps are added to  make sure that at least $2f+1$ replicas are aware of view change. And $(3)$ faulty nodes ($f<n/3$) cannot trigger a view change. \looseness=-1

It should be noted that even during a view change the performance of Hermes will not be affected by the slack primary. This is because the primary will send its message to the impetus committee and the impetus committee will forward primary's message to regular nodes. Therefore view change messages from primary will propagate proportional to the wire speed of prompt nodes. \looseness=-1

\section{Efficient Optimistic Responsiveness}
\label{Optimization}
Pipelining is an optimization technique that involves sending  requests/messages back to back without waiting for their confirmation. This technique has been previously used in networking \cite{PipeliningInHTTP}  as well as in consensus \cite{Tuning-Paxos-for-High-Throughput-with-Batching-and-Pipelining, Hot-stuff} to improve performance.
Therefore, as optimization the primary can propose a new block $B'$ with sequence $s$  after $Pre-Proposal$ phase for block $B$ with sequence $s-1$ without waiting for the block $B$ to get committed. The primary can wait only for majority of votes from the impetus committee to propose a new block. This technique decreases the wait time for the primary to propose next block while making sure that with high probability the protocol will make progress.\looseness=-1

For safety condition to hold we need to replace $\beta$  with $A_{\beta}$, where  $A_{\beta}$ is an array of $\beta$ in $ViewChange$ message $V$. As in basic Hermes the primary proposes the next block after it commits its parent block. There will be only one valid pending block in the network with $\beta$. But in case of pipelined Hermes, Since the primary proposes a block after $Pre-Proposal$ phase of its parent, by the end of timeout a node might receive multiple uncommitted blocks. Therefore, the information about latest uncommitted blocks has to be included in $A_{\beta}$ during the view change. \looseness=-1 

For state machine replication (SMR) in blockchains sending $Pre-proposal$ message, $Confirm$ message and  committing a block during $Approval$ phase has to be done in order (serially) for each block. 
This means a node $i$ will not confirm  block $B$ at the sequence $s$ ($\eta_i(B,s) \gets true$) before  block $B'$ with the sequence $s'$ such that $s>s'$. Out of order messages can be cached.
But other operations  including signature verification and format checking can be done concurrently.

We leave more details on this to programmers who will be going to implement this protocol. \looseness=-1

\begin{figure*}
\hfill
    \begin{minipage}{0.45\linewidth}
       \includegraphics[width=6cm,height= 4cm]{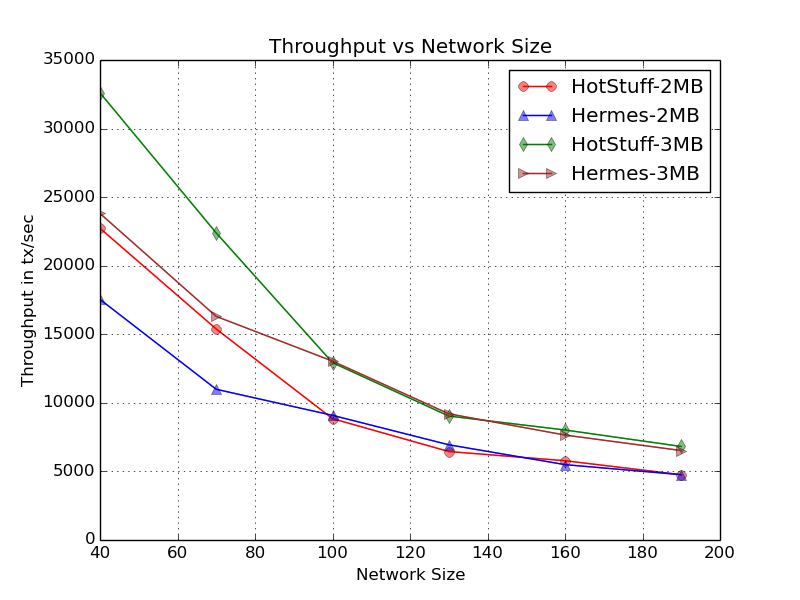}
        \caption{\protect\oalign{Throughput with different network sizes when \hfil\strut\cr\strut all nodes  have higher bandwidth \hfil}}
     \label{Evaluation:NormalThroughputTests}

           \end{minipage}
           \hfill
\begin{minipage}{0.45\linewidth}
    \includegraphics[width=6cm, height= 4cm]{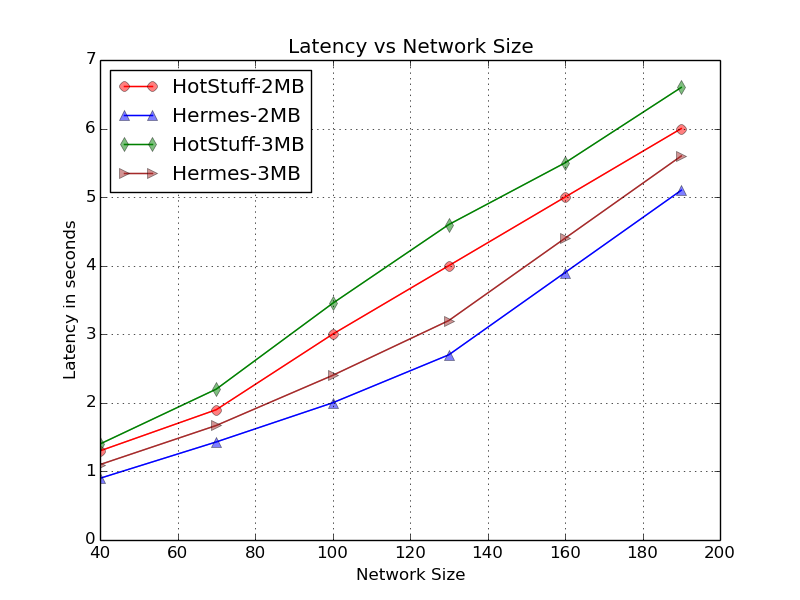}
   \caption{\protect\oalign{Latency with different network sizes when \strut\cr\strut all  nodes  have higher bandwidth \hfil}}
   \label{Evaluation:NormalLatencyTests}
    \end{minipage}


\end{figure*}

\begin{figure*}
\hfill
    \begin{minipage}{0.45\linewidth}

       \includegraphics[width=6cm, height= 4cm]{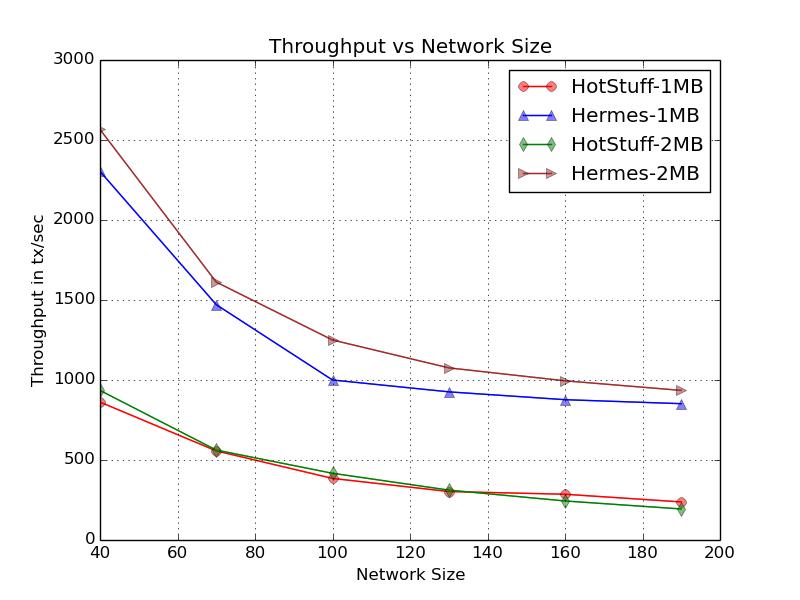} 
        \caption{Throughput when a slack primary is introduced}
            \label{Evaluation:LowThroughputTests}
    
               \end{minipage}
               \hfill
\begin{minipage}{0.45\linewidth}
     \includegraphics[width=6cm, height= 4cm]{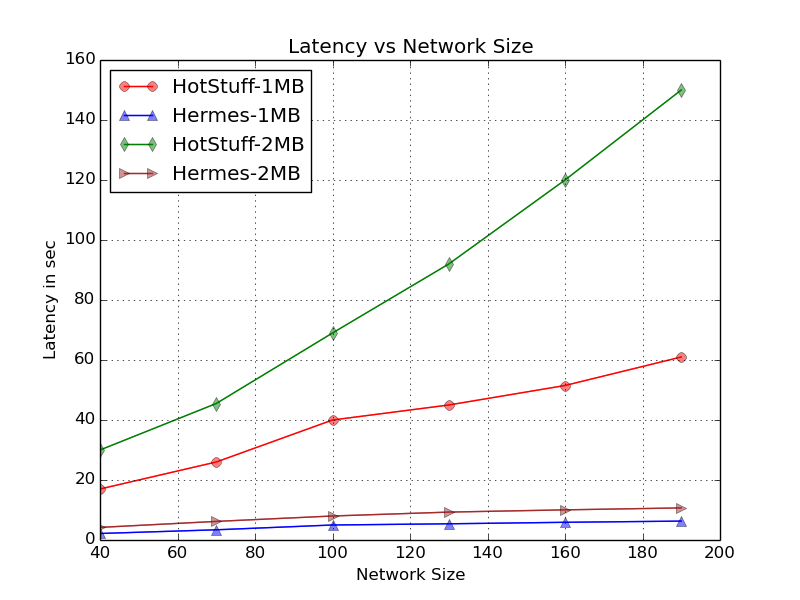} 
    \caption{Latency when a slack primary is introduced}
    \label{Evaluation:LowLatencyTests}
        \end{minipage}

\end{figure*}

\section{Evaluations}
\label{section:experiments}

We have implemented prototype of optimized Hermes as well as chained Hotstuff \cite{Hot-stuff} (pipelined) using Golang programming language.
We chose pipelined Hotstuff as it is a state-of-the-art consensus protocol. A variant of Hotstuff called Libra-BFT \cite{libra-BFT} is also used by Facebook. We implemented Hotstuff to keep similar code architecture with Hermes so that we can fairly compare both algorithms.

Experiments were conducted in the Amazon Web Services (AWS) cloud. Each node in the network was represented by an instance of type {\em t2.large} in AWS. Each {\em t2.large} instance has 2 virtual CPU (cores) and 8GB of memory. We recorded performance of Hermes as well as Hotsuff with  network sizes of $40$, $70$, $100$, $130$, $160$ and $190$ nodes. The impetus committee size $c$ for each network size $n$ are $18$, $27$, $30$, $33$, $34$ and $36$ respectively. We adjusted values of $n$ and $c$ so that we can obtain the maximum failure probability of $P_f\leq 3.8 \times 10^{-22}$ using Equation \ref{Eq:Failure Probability Equation}. 

Randomly generated transactions were used during experiments to transfer funds among different accounts. Transactions were included in a block (batch). We used different block sizes of 1MB, 2MB and 3MB. 

During our first evaluation all nodes had similar upload and download  bandwidth of $1$Gbps $128$MBps. During our tests as shown in Figure \ref{Evaluation:NormalThroughputTests} Hotstuff performs better with small network size ($40$ and $70$). This is due to the reason that when $n$ is small, $c$ ($c < n$) has to be larger to maintain a safe failure probability.  
Therefore, as the network size $n$ grows larger the growth of $c$ remains sub linear to the growth $n$ while keeping safe failure probability. As a result, with larger $n$ Hermes performs better and its performance is comparable to the performance of HotStuff. 
Whereas, in terms of latency as shown in Figure \ref{Evaluation:NormalLatencyTests}, Hermes  outperforms the Hotstuff. As it can be seen, Hermes has lower latency than the Hotstuff and variation in the latency of Hermes and Hotstuff with different block sizes of $2MB$ and $3MB$
grow larger when $n>70$.\looseness=-1

Next we introduced slack primaries and then compared the performance of Hermes and Hotstuff as shown in Figure \ref{Evaluation:LowThroughputTests} and Figure \ref{Evaluation:LowLatencyTests}. We configured the upload speed of slack primary to $82$Mbps (roughly equivalent to $10$MBps). We reduced block size to $1$ MB and $2$ MB 
as with block size of  $3$MB
we were experiencing very high latency and packet drops.
As it can be seen in Figure \ref{Evaluation:LowThroughputTests} the throughput lead of Hermes grows from more than two times (when $n=40$) to four times when block size is $1$MB and to $4.8$ times when block size is $2$MB (with $n=190$). Similarly, Figure \ref{Evaluation:LowLatencyTests} shows that the latency of Hotstuff grows very fast 
with increase in network size when primary is slack.

\section{Related Work}
\label{Section: Related Work}
The most relevant work to our protocol is Proteus  \cite{Jala1907:Proteus} in which the primary proposes  block to a subset of nodes called root committee.  But after confirmation from the root committee the primary proposes the block to all nodes which we have avoided in Hermes. Proteus has only focused on improving the message complexity.
Moreover, the probability of view change in Proteus is higher ($0.5$) compared to Hermes ($0.3$) as the primary is chosen from the impetus committee. 
Proteus is only R-Safe whereas Hermes is R-safe and highly likely S-safe. To be S-safe Hermes has additional recovery phase during its view change sub protocol. Proteus holds responsiveness but is not efficiently optimistic responsive whereas Hermes is efficiently optimistic responsive.\looseness=-1 

Random selection of subset of nodes in a BFT committee has also been used inAlgorand \cite{Algorand}, Consensus Through Herding \cite{ConsensusTH}, \cite{ SublinearRBFT} and \cite { Communication-Complexity}. Algorand is a highly scalable BFT based protocol, but can tolerate $20$ percent of Byzantine nodes in the network  while providing probabilistic safety  (depending  on  the  size  of  committee  and  number  of Byzantine nodes in the network).  Consensus Through Herding \cite{ConsensusTH} achieve consensus in poly-logarithmic number of rounds. In \cite{ SublinearRBFT,Communication-Complexity}  protocols achieves Byzantine Agreement (BA) with sublinear round complexity under corrupt majority. 
All of these protocols operate in Synchronous environment under adaptive adversarial model.  
On the contrary, Hermes can operate in asynchronous environment and does not consider adaptive adversarial model.

Hot-stuff \cite{Hot-stuff} is another BFT-based protocol linear message complexity $O(n)$. Hot-stuff has abolished the view change mode with the cost of an additional round of message complexity in normal mode.
In Hot-stuff there can be multiple forks before the block gets committed. When a block gets committed, the blocks in its competitive forks will be revoked. This effectively causes wastage of resources (bandwidth, processing power etc.) that were consumed in proposing the revoked blocks.
As the primary is changed with each block proposal, at most there can be $\lfloor n/3 \rfloor -1$ failed epochs (where no blocks are generated) out of $n$ epochs. The rotating primary mechanism of Hotstuff cannot resolve the problem of slack primary. Because if there are $n_s$ slack nodes. Then at least $n_s$ times there will be sudden degradation in protocol performance due to slack primary. This results in inconsistent performance.
Libra-BFT is \cite{libra-BFT} based on Hotstuff, but  its  synchronization module that is used to recover missing blocks has $O(n^2)$ message complexity.

There has been another line of work presented in PRIME \cite{Prime}, Aavardak\cite{Clement:2009:MBF:1558977.1558988} and RBFT \cite{Aublin:2013:RRB:2549695.2549742}. In these papers the authors have tried to address the performance attack (delaying proposal) by Byzantine primary by monitoring primary performance. In case of failing the expectation, primary is replaced by the network. But in their model, they have not considered the slack but honest nodes. That means slack primaries will also be treated as a Byzantine primary and will be replaced. Moreover, even a prompt primary might get replaced due to network glitch.  This will result in higher number of expensive but unnecessary view changes. Moreover, if Proof-of-Stake (PoS) is used over BFT, then an honest but slack primary will seldom get chance to propose a block, hence will not be able to increase its stake regularly causing centralization of stake to prompt nodes. Furthermore, unlike Hermes these solutions do not address combination of problems which include message complexity, latency, and slack primary problem. 
Hermes does not consider addressing the Byzantine primary performance attack but performance monitoring module from these proposal can be added to Hermes to address this issue.

\section{Conclusion}
\label{Conclusion}
In this paper we proposed a highly efficient BFT-based consensus protocol for blockchains which addresses the problem of slack primary. Furthermore, this protocol has lower latency and message complexity. We also show that breaking strong safety in Hermes is unlikely and expensive for Byzantine primary. Breaking strong safety will only cause additional processing delay and will not result in double spending attack. Therefore, a Byzantine primary has no incentive to perform such attack.  As a result the safety of Hermes is comparable to the safety of general BFT-based protocols.

\bibliographystyle{IEEEtran}
\bibliography{OPBFT,TBFT}

\end{document}